\documentclass[conference]{IEEEtran}

\usepackage[utf8]{inputenc}
\usepackage{verbatim}
\usepackage{optidef}

\usepackage{amsmath,amssymb,amsthm,amsfonts,mathrsfs}
\usepackage{tablefootnote}
\usepackage[margin=0.5in]{geometry}

\usepackage{hyperref} 

\usepackage{graphicx,color,cases} 
\usepackage{xfrac}  
\usepackage{xargs}

\usepackage[colorinlistoftodos,prependcaption]{todonotes}

\definecolor{DarkGreen}{rgb}{0.1,0.5,0.1}
\definecolor{DarkRed}{rgb}{0.5,0.1,0.1}
\definecolor{DarkBlue}{rgb}{0.1,0.1,0.5} 
  
\def\>{\rangle} 
\def\<{\langle}
\def\anc{|{\rm ancilla}\>}
\def\banc{\<{\rm ancilla}|}

\newtheorem{definitionenv}{Definition}
\newtheorem{lemmaenv}[definitionenv]{Lemma}
\newtheorem{theoremenv}[definitionenv]{Theorem}
\newtheorem{corollaryenv}[definitionenv]{Corollary}
\newtheorem{propositionenv}[definitionenv]{Proposition}
\newtheorem{factenv}[definitionenv]{Fact}
\newtheorem{conjectureenv}[definitionenv]{Conjecture}
\newtheorem{exampleenv}{Example}
\newtheorem{app-lemmaenv}[section]{Lemma}

\newenvironment{definition}{\begin{definitionenv}\rm}{\end{definitionenv}}
\newenvironment{lemma}{\begin{lemmaenv}\rm}{\end{lemmaenv}}
\newenvironment{theorem}{\begin{theoremenv}\rm}{\end{theoremenv}}
\newenvironment{corollary}{\begin{corollaryenv}\rm}{\end{corollaryenv}}
\newenvironment{example}{\begin{exampleenv}\rm}{\end{exampleenv}}
\newenvironment{fact}{\begin{factenv}\rm}{\end{factenv}}
\newenvironment{proposition}{\begin{propositionenv}\rm}{\end{propositionenv}}
\newenvironment{conjecture}{\begin{conjectureenv}\rm}{\end{conjectureenv}}
\newenvironment{app-lemma}{\begin{app-lemmaenv}\rm}{\end{app-lemmaenv}}

\newcommand{\bd}{\begin{definition}}
\newcommand{\ed}{\end{definition}}
\newcommand{\bl}{\begin{lemma}}
\newcommand{\el}{\end{lemma}}
\newcommand{\elp}{\hspace*{\fill} $\Box$
                 \end{lemma}}
\newcommand{\bt}{\begin{theorem}}
\newcommand{\et}{\end{theorem}}
\newcommand{\etp}{\hspace*{\fill} $\Box$
                 \end{theorem}}
\newcommand{\bc}{\begin{corollary}}
\newcommand{\ec}{\end{corollary}}
\newcommand{\ecp}{\hspace*{\fill} $\Box$
                 \end{corollary}}
\newcommand{\bcj}{\begin{conjecture}}
\newcommand{\ecj}{\end{conjecture}}

\newcommand{\be}{\begin{example}}
\newcommand{\ee}{\end{example}}
\newcommand{\eep}{\hspace*{\fill} $\Box$
                 \end{example}}
\newcommand{\bp}{\begin{proposition}}
\newcommand{\ep}{\end{proposition}}
\newcommand{\epp}{
                 \end{proposition}}

\newcommand{\tr}{\text{Tr}}

\newcommandx{\yellownote}[2][1=]{\todo[linecolor=yellow,backgroundcolor=yellow!25,bordercolor=yellow,#1]{#2}}
\newcommandx{\greennote}[2][1=]{\todo[inline,linecolor=olive,backgroundcolor=green!25,bordercolor=olive,#1]{#2}}

\begin{document}
\title{Permutation-invariant quantum coding for quantum deletion channels}
\author{%
  \IEEEauthorblockN{Yingkai Ouyang}
  \IEEEauthorblockA{
Department of Electrical and Computer Engineering,\\
   National University of Singapore\\
   Email:  oyingkai@gmail.com}
}

\maketitle

\begin{abstract}
Quantum deletions, which are harder to correct than erasure errors, occur in many realistic settings. It is therefore pertinent to develop quantum coding schemes for quantum deletion channels. To date, not much is known about which explicit quantum error correction codes can combat quantum deletions. We note that {\em any} permutation-invariant quantum code that has a distance of $t+1$ can correct $t$ quantum deletions for any positive integer $t$ in both the qubit and the qudit setting. Leveraging on coding properties of permutation-invariant quantum codes under erasure errors, we derive corresponding coding bounds for permutation-invariant quantum codes under quantum deletions. We focus our attention on a specific family of $N$-qubit permutation-invariant quantum codes, which we call shifted gnu codes. The main result of this work is that their encoding and decoding algorithms can be performed in $O(N)$ and $O(N^2)$. 
\end{abstract}


\section{Introduction}

Quantum states are fragile and are extremely susceptible to noise. The type of noise encountered by quantum states depends very much on the underlying mechanism for decoherence. In the problem of quantum communication, quantum states are transmitted over a lossy quantum channel, and it is natural to expect that during this imperfect communication process, some of underlying qubits will be inadvertently lost.
This can arise from temporarily blockage, misalignment, or destruction of the transmitted signal.
When we know which of the transmitted qubits have been lost, we would have a quantum erasure channel \cite{bennett1997capacities,grassl1997codes},
from which error correction can be implemented based on the distance of the quantum code. However in realistic scenarios, we often do not know which qubits are lost, and this results in so-called quantum deletions, which is the focus of this paper.

Decades since Levenshtein first developed classical codes to combat deletions \cite{levenshtein1966binary}, there continues to be much interest in developing classical deletion codes \cite{WachterZeh2017TIT,kdeletion2020}.
In contrast, quantum coding schemes for deletions only been considered recently \cite{leahy2019quantum,nakayama2020first,HagiwaraISIT2020,nakayama2020single}.
In \cite{HagiwaraISIT2020}, the authors propose using a permutation-invariant four-qubit code for correcting for a single deletion error, and give explicit encoding and decoding circuits for it, and prove its optimality with respect to the number of qubits used. 
While the authors later generalized their theory for quantum coding for more general quantum codes to correct single deletions \cite{nakayama2020single}, extending their scheme to correct multiple deletions remains an open problem. 
In \cite{leahy2019quantum}, instead of directly considering quantum coding for deletions, the authors propose to use the methods of channel simulation to map deletions to erasures, and use conventional quantum codes thereafter. While this channel simulation method allows the correction of multiple deletions, it incurs additional communication overheads, and may not be robust against additional types of noise, such as amplitude damping errors which model energy loss in physical systems. 
The problem of designing quantum coding schemes that pertain directly to multiple deletions remains to be fully addressed.

In this paper, we propose to leverage on the existing theory of permutation-invariant quantum codes \cite{Rus00,PoR04,ouyang2014permutation,ouyang2015permutation,OUYANG201743,ouyang2019permutation} to address the problem of quantum coding for quantum deletions.
Permutation-invariant quantum codes are invariant under {\em any} permutation of the underlying particles, and have been studied in the context of quantum storage within the ground space of Heisenberg ferromagnets \cite{ouyang2019mems}, for robust quantum metrology \cite{ouyang2019robust}, 
and also for decoherence-free quantum communication on bosonic quantum buses \cite{ouyang2019permutation}.
Remarkably, for permutation-invariant quantum codes, deletions  are equivalent to erasures.
Hence it readily follows that {\em any} permutation-invariant quantum code that corrects $t$ erasures can also correct $t$ deletions. 

The main contribution of this work addresses the complete coding problem for permutation-invariant quantum codes with respect to deletion channels, and goes beyond code constructions. Indeed with just code constructions for permutation-invariant quantum codes, the coding problem remains incomplete without a full treatment of the encoding and decoding procedure. Aside for quantum coding schemes for single-deletions \cite{HagiwaraISIT2020,nakayama2020single}, not much is known for the general multiple deletion scenario.
To date, quantum encoding and decoding for permutation-invariant codes remains to be fully addressed. 

This paper aims to give a complete solution to the quantum coding problem for quantum deletion channels using permutation-invariant quantum codes. 
We achieve this by first showing how existing results related to the efficient preparation of permutation-invariant states can be combined with gate-teleportation to provide efficient encoding schemes for permutation-invariant codes.
Second, we derive efficient decoding algorithms for permutation-invariant codes capable of correcting multiple deletions.
The complete quantum coding scheme we propose will not only be advantageous to realize practical schemes to combat quantum deletions, but also help accelerate the development of quantum technologies based upon permutation-invariant quantum codes, such as quantum memories \cite{ouyang2019mems} and quantum sensors \cite{ouyang2019robust}.

Before we elaborate the details of the encoding and decoding permutation-invariant quantum codes, 
we briefly review the relevant theory of permutation-invariant quantum codes in Section \ref{sec:picode}, 
and review a particular family of such codes, known as gnu codes \cite{ouyang2014permutation}.
We consider a variation of these gnu codes derived in \cite{OUYANG201743}, which we call shifted gnu codes. 
We also leverage on coding properties of permutation-invariant quantum codes under erasures to derive bounds relating the number of correctible deletions $t$ and the length of quantum codes that directly correct quantum deletions. 
In Section \ref{sec:encoding} we show that the encoding of certain permutation-invariant quantum codes can be done in linear time.
In Section \ref{sec:decoding},
we explain how the decoding of shifted gnu codes can be done after $t$ deletions occur. This decoding procedure comprises of (1) quantum phase estimation, (2) probability amplitude rebalancing, and (3) gate teleportation, and can be achieved in quadratic time. 

\section{Permutation-invariant quantum codes and deletions}
\label{sec:picode}

Permutation-invariant quantum codes, by virtue of their symmetry, must reside within the symmetric subspace of any multi-particle system.
When each particle is a qudit of dimension $q$, the dimension of this symmetric subspace is $\binom {n+q-1} {q-1}$.
As mentioned in the introduction, for permutation-invariant quantum codes, deletions are equivalent to erasures. 
It therefore suffices to consider the usual notion of distance when considering the correction of deletion errors on permutation-invariant codes. We emphasize this fact in the following statement.
\begin{fact}
Having $t$ deletion errors is equivalent to having $t$ erasure errors on any permutation-invariant quantum code.
\end{fact}
\begin{proof}
Deletion errors occur when particles are lost and we do not know which particles are lost. Erasures occur when we lose particles, but do not know which particles are lost. Since permutation-invariant quantum codes are invariant under any permutation of the underlying particles, deletion of $t$ unknown particles is equivalent to the erasure of the first $t$ particles for permutation-invariant quantum codes, which proves Fact 1.
\end{proof}
From this, we can derive coding bounds for deletion channels by leveraging on known coding bounds on permutation-invariant quantum codes.

\begin{theorem}
\label{thm:codingtheorem}
For any $q \ge 2$, there are $q$-ary $N$ qudit quantum codes with $M$ logical codewords and correcting $t$ deletions whenever
\begin{align}
N \ge (t+1)^2 (M-1). \label{eq:codingbound}
\end{align}
\end{theorem}
\begin{proof}
This follows directly from the constructed permutation-invariant quantum codes' ability to correct erasure errors in Ref. \cite[Theorem 6.7]{OUYANG201743} In particular, Ref. \cite[Theorem 5.2]{OUYANG201743} gives an explicit way to construct these codes for every $N$ that satisfies the inequality.
\end{proof}
When $t=1$ and $M=2$, the inequality \eqref{eq:codingbound} in Theorem \ref{thm:codingtheorem} is equivalent to $N \ge 4$. Since it was shown that the single-deletion permutation-invariant quantum code on four-qubits is optimal \cite{HagiwaraISIT2020}, it follows that the inequality \eqref{eq:codingbound} is tight for $t=1$ and $M=2$.

To correct two deletions, we can use the Pollatsek-Ruskai 7-qubit permutation-invariant quantum code \cite{PoR04} which has a distance of three. Hence for $t=2$, the inequality  \eqref{eq:codingbound} is not tight.

\subsection{Basis states}
To describe permutation-invariant quantum states on qubit systems, we use the Dicke basis comprising of quantum states called Dicke states.
Each Dicke state on $N$ qubits and of weight $w$ is a uniform superposition over classical codewords with a constant Hamming weight equal to $w$, and has the form
\begin{align}
|D^N_w\> = \frac{1}{\sqrt{\binom N w}} \sum_{\substack{ x_1,\dots, x_N \in \{0,1\} \\ x_1+\dots + x_N = w }} |x_1\>\otimes \dots \otimes |x_N\>.\label{eq:dicke-defi}
\end{align}
The Dicke states form an orthonormal basis of the qubit symmetric subspace, and satisfy the orthonormality conditions $\<D^N_u |D^N_v\> = \delta_{u,v},$ where $\delta_{u,v}$ denotes the Kronecker delta function. 
It is often convenient when working with symmetric states to work with the unnnormalized Dicke states given by
\begin{align}
|H^N_w\> = \sum_{\substack{ x_1,\dots, x_N \in \{0,1\} \\ x_1+\dots + x_N = w }} |x_1\>\otimes \dots \otimes |x_N\>.\label{eq:unnormalized-dicke-defi}
\end{align}
Clearly, we have $|H^N_w\> = \sqrt{\binom N w}|D^N_w\>$.
By applying the Vandermonde identity,
for any non-negative integer $t$ such that $t \le w\le N-t$, we clearly have
\begin{align}
        |H^N_{w}\> &= 
    \sum_{s=0}^{t} 
    |H^{t}_{s}\> \otimes |H^{N-t}_{w-s}\> . \label{eq:vandermonde}
\end{align}
This decomposition helps us unravel the structure of permutation-invariant quantum codes after they experience deletions.

\subsection{Gnu codes}
In Ref \cite{ouyang2014permutation},
a family of permutation-invariant quantum codes with three integer parameters $g,n$ and $u$ was introduced.
Intuitively, $g$ and $n$ correspond to the bit-flip and phase-flip distances respectively, and $u\ge 1$ is a scaling parameter that is independent of the code's distance.
More precisely, the distance of gnu codes is given by $\min\{g,n\}$.
These codes are called gnu codes because they comprise of $N = g n u$ qubits.

The corresponding logical codewords are 
\begin{align}
|0_L\> = \sqrt{2^{-(n-1)}} \sum_{\substack{0 \le j \le n \\ j\ {\rm even}} } 
\sqrt{\binom n j}
|D^{gnu}_{gj}\>	 ,\notag\\
|1_L\> = \sqrt{2^{-(n-1)}} \sum_{\substack{0 \le j \le n \\ j\ {\rm odd}} } 
\sqrt{\binom n j}
|D^{gnu}_{gj}\>.
\end{align}
 
Note that the optimal four-qubit permutation-invariant code proposed in Ref. \cite{HagiwaraISIT2020} that corrects a single deletion is a gnu code with $g=n=2$ and $u=1,$ because its logical codewords are
\begin{align}
|0_L\> &= \frac{|0000\>+|1111\>}{\sqrt 2}	 ,\quad
|1_L\> = |D^4_2\>.
\end{align}
To correct two deletion errors, we can use the nine-qubit Ruskai code \cite{Rus00} which has $g=n=3$ and $u=1$, and logical codewords
\begin{align}
|0_L\> &= 
\frac{|D^9_0\>+\sqrt 3|D^9_6\>}{2}	 ,\quad
|1_L\> = \frac{ \sqrt 3 |D^9_3\>+|D^9_9\>}{2} ,
\end{align}
In general, to encode one logical qubit into $N$ qubits while correcting $t$ deletions, it suffices to use gnu codes of length $N=(t+1)^2$, with $g=n=t+1$ and $u \ge 1$. 
Such gnu codes with $u=1$ saturate the bound in \eqref{eq:codingbound}.

When $u=1$ and $n$ is odd, the transversal operator $X^{\otimes N}$ is a logical bit-flip operation on the corresponding gnu code, where $X= |0\>\<1|+|1\>\<0|$ denotes the qubit bit-flip operator.

\subsection{Shifted gnu codes}

In Ref. \cite{OUYANG201743}, gnu codes were generalized to admit logical codewords of the form 
\begin{align}
|0_L\> = \sqrt{2^{-(n-1)}} \sum_{\substack{0 \le j \le n \\ j\ {\rm even}} } 
\sqrt{\binom n j}
|D^{gnu+\Delta}_{gj+\Delta}\>	 ,\notag\\
|1_L\> = \sqrt{2^{-(n-1)}} \sum_{\substack{0 \le j \le n \\ j\ {\rm odd}} } 
\sqrt{\binom n j}
|D^{gnu+\Delta}_{gj+\Delta}\>,
\end{align}
where the non-negative integer $\Delta$ denotes a shift.
We call these $\Delta$-shifted gnu codes.
In Ref. \cite{OUYANG201743}, it was shown that these $\Delta$-shifted gnu codes have 
a minimum distance of $\min\{g,n\}$ for any non-negative shift $\Delta$ and $u\ge 1$.

In this paper, we will restrict our attention to the quantum coding problem for shifted gnu codes for which 
$g=\Delta$, $n=2 \lfloor{g/2}\rfloor+1$, and $u=1+ 1/n$ so that 
$N = gn + 2g$.
These shifted gnu codes have logical codewords given by
\begin{align}
|0_L\> = \sqrt{2^{-(n-1)}} \sum_{\substack{0 \le j \le n \\ j\ {\rm even}} } 
\sqrt{\binom n j}
|D^{gn+2g}_{gj+g}\>	 ,\notag\\
|1_L\> = \sqrt{2^{-(n-1)}} \sum_{\substack{0 \le j \le n \\ j\ {\rm odd}} } 
\sqrt{\binom n j}
|D^{gn+2g}_{gj+g}\>.\label{eq:special-shifted-gnu}
\end{align}
The shifted gnu codes in \eqref{eq:special-shifted-gnu} also have a logical bit-flip operator that can be implemented transversally by the operator
$X^{\otimes N}$. This is because $\binom n j = \binom n {n-j}$ combined with the fact that 
\begin{align}
X^{\otimes N} |D^{gn+2g}_{gj+g}\>
=|D^{gn+2g}_{g(n - j)+g}\>
\end{align}
and that $n-j$ has a different parity from $j$ for odd $n$.
The logical $Z = |0\>\<0| - |1\>\<1|$ operator can be implemented by $R^{\otimes N}$ where
\begin{align}
R = |0\>\<0| + e^{\pi i /g} |1\>\<1|,
\end{align}
because 
\begin{align}
R^{\otimes N}|D^N_{gj + g}\> &= |D^N_{gj + g}\>,& \mbox{$j$ even}\\
R^{\otimes N}|D^N_{gj + g}\> &=- |D^N_{gj + g}\>,& \mbox{$j$ odd	}.
\end{align}

An explicit example of a shifted gnu code of the form of \eqref{eq:special-shifted-gnu} that corrects two deletions is a variant of the Ruskai code, which is a 15-qubit code with logical codewords 
\begin{align}
|0_L\> =  \frac{|D^{15}_{3}\> +\sqrt 3 |D^{15}_{9}\>}{2} 	 ,\quad
|1_L\> = \frac{\sqrt 3|D^{15}_{6}\> + |D^{15}_{12}\>}{2} 	 .
\end{align}
In this paper, we focus on the quantum coding problem for shifted-gnu codes of the form of \eqref{eq:special-shifted-gnu}.

\section{Encoding schemes}
\label{sec:encoding}

Here, we treat the problem of encoding an arbitrary qubit 
$a|0\> + b|1\>$ into a logical qubit 
$a|0_L\> + b|1_L\>$ of gnu codes and shifted gnu codes that have $X^{\otimes N}$ as their logical bit-flip operator.

The first step of the encoding scheme is to prepare the state $|0_L\>$. Since $|0_L\>$ is a permutation-invariant quantum state, we can use {\em any} scheme for preparing permutation-invariant states, which can be done for instance in the quantum circuit model in $O(N^2)$ steps \cite{bartschi2019deterministic},
using superconducting charge qubits in a ultrastrong coupling regime in $O(N)$ steps \cite{init-picode-2019-PRA}, 
and also using geometric phase gates in $O(N)$ steps \cite{johnsson2020geometric}.

Once the logical codeword $|0_L\>$ is prepared, we can use a logical version of gate teleportation \cite[Eq. (5)]{ZLC00} as shown in Fig \ref{fig:encoding}.
\begin{figure}[htp]
\centering 
\includegraphics[width=0.35\textwidth]{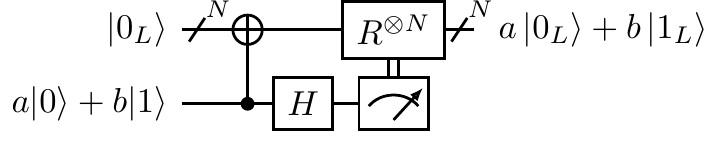}  
\caption{A quantum circuit that teleports an unencoded qubit $a|0\> + b|1\>$ to a logical qubit $a|0_L\> + b|1_L\>$ in a gnu code or a shifted gnu code.
For this scheme to work, it suffices to have $X^{\otimes N}$ to implement a logical bit-flip on the code.} 
\label{fig:encoding}
\end{figure} 
In this gate teleportation protocol, the CNOT from the control encoding qubit to the permutation-invariant quantum state applies $N$ two-qubit CNOT gates, where each two-qubit CNOT gate has its control on the encoding qubit and its target on each qubit in the permutation-invariant quantum state.
Here $H = (X + Z)/\sqrt 2$ denotes the Hadamard gate on a single qubit, where 
$Z = |0\>\<0|  - |1\>\<1|$ denotes a qubit phase-flip operation,
and the measurement of the control basis is done in the basis $\{|0\>,|1\>\}$.
If the measurement obtains $|0\>$ as its outcome, nothing is done on the permutation-invariant state. Otherwise, the transversal gate $R^{\otimes N}$ is applied on the permutation-invariant quantum state.
The output of the quantum circuit in Fig \ref{fig:encoding} yields the desired logical state
$a|0_L\> + b|1_L\>$.

The computation cost of implementing Fig \ref{fig:encoding} is linear in $N$. Namely, it requires $N$ two-qubit CNOT gates, one single-qubit Hadamard gate, one measurement of a single-qubit and $N$ applications of $R$ to give a total gate count of $2N+2$.
Since the state preparation of $|0_L\>$ can be done in $O(N)$ steps, we arrive at the following result.
\begin{theorem}
The encoding circuits of gnu codes and shifted gnu codes with $X^{\otimes N}$ as their logical bit-flip operator can be encoded in $O(N)$ steps.
\end{theorem}

\section{Decoding}
\label{sec:decoding}

If a quantum code satisfies the celebrated Knill and Laflamme quantum error correction criterion \cite{KnL97} exactly, we can use the procedure outlined in \cite{KnL97} and \cite{LNCY97} to construct explicit recovery maps for the quantum code. 
Furthermore, these recovery maps can be decomposed into two steps.
In the first step, a projective measurement is done, projecting the erroneous quantum state into disjoint correctible subspaces.
In the second step, a conditional unitary operation brings the correctible subspace back to the original subspace.
While this Knill-Laflamme recovery procedure gives formal treatment of quantum error correction in the language of matrix algebra, it does not tell us how we may realize these operations in practice.

In this section, we show that decoding shifted gnu codes can be achieved in polynomial time.
Before we describe how the decoding scheme proceeds, let us examine what happens after an initial quantum state $|\psi\> = a|0_L\> + b|1_L\>$ encounters a quantum deletion / erasure channel $\mathcal D_t$.
As an initial step, we let $\tr_t$ denote the partial trace of the first $t$ particles of an $N$ qudit state, and evaluate $\tr_t(|H^N_{gj}\>\<H^N_{gk}|)$ in the following lemma, reminiscent of \cite[Eq (III.4)]{ouyang2019robust}.
\begin{lemma}
\label{lem:intermediate-partial-trace}
Let $t$ be non-negative integer and $N$ be a positive integer. Let $u$ and $v$ be integers such that $t \le u,v \le N-t$.
Then for any $N$-qudit operator $|H^N_{u}\>\<H^N_{v}|$, we have 
\begin{align}
        \tr_t \left( |H^N_{u}\>\<H^N_{v}| \right) &= 
    \sum_{s=0}^t 
 \binom t s    |H^{N-t}_{u-s}\>  \<H^{N-t}_{v-s}|. 
\end{align}
\end{lemma}
\begin{proof}
Under the assumption of the lemma,
we can use \eqref{eq:vandermonde} to get
\begin{align}
    |H^N_{u}\> = 
    \sum_{s=0}^t 
    |H^t_s\>
    \otimes 
    |H^{N-t}_{u-s}\>.\notag
\end{align}
Now we take the partial trace on the first $t$ qubits of the operator 
$|H^N_{u}\>\<H^N_{v}|$.
Using the cyclic property of the trace on the first $t$ qubits, the operator becomes
\begin{align}
&
\sum_{s,s'=0}^t 
    \<H^t_{s'}|H^t_s\>
    |H^{N-t}_{u-s}\>
    \<H^{N-t}_{v-{s'}}|=
\sum_{s=0}^t 
\binom {t}{s}
    |H^{N-t}_{u-s}\>
    \<H^{N-t}_{v-s}|.\notag
\end{align} 
\end{proof}
Now note that for any permutation-invariant quantum state $|\psi\>$,
we have $\mathcal D_t(|\psi\>\<\psi|) = \tr_t (|\psi\>\<\psi|)$.
Furthermore, we find that 
when $|\psi\>$ resides in the codespace of an appropriate shifted gnu code, $\mathcal D_t(|\psi\>\<\psi|)$ admits the following spectral decomposition,
as shown in the following lemma.
\begin{lemma}
\label{lem:state-t-deletions}
Consider any $N$-qubit $\Delta$-shifted gnu code with parameters $g,n$ and $\Delta$ with logical codewords $|0_L\>$ and $|1_L\>$.
Let $t$ denote the number of deletions, and suppose that $t \le \Delta$ and $gn \le N-t$. 
Define the unnormalized quantum states
\begin{align}
|\pm_s\> = 
\sum_{j=0}^n (-1)^{j}
b_{j,s} |D^{N-t}_{gj+\Delta-s}\> ,
\end{align}
where
$b_{j,s} = \sqrt{
\binom n j 
\binom {N-t}{gj+\Delta-s}}
/\sqrt{2^n\binom {N}{gj+\Delta}}$.

Then for any $|\psi\> = \alpha|+_L\> + \beta|-_L\>$ in the codespace,
where $|\pm_L\> = (|0_L\> \pm |1_L\>)/\sqrt 2$,
after $t$ deletions, the state $|\psi\>\<\psi|$ becomes
\begin{align}
    D_t(|\psi\>\<\psi|)
    =
    \sum_{s=0}^t p_s |\phi_s\>\<\phi_s| \label{eq:deletion-convex-combination},
\end{align} 
where $p_s = \binom t s/ \<+_s|+_s\>$,
$|\phi_s\> = |\psi_s\>/\sqrt{\<+_s|+_s\>}$ and $|\psi_{s}\> = 
\alpha |+_s\> + \beta|-_s\>$.
\end{lemma}
\begin{proof}
By definition of shifted gnu states, we can write 
\begin{align}
    |\psi\>
    &= 
\sum_{j=0}^n
\frac{\sqrt{\binom n j}}{\sqrt{2^n}}
(\alpha + (-1)^j\beta )|D^{N}_{gj+\Delta}\>
\notag\\
&=
\sum_{j=0}^n
a_j
(\alpha + (-1)^j\beta )|H^{N}_{gj+\Delta}\>,
\end{align}
where $a_j = \sqrt{\binom n j}/\sqrt{2^n\binom {N}{gj+\Delta}}$.
Then we can write
\begin{align}
    |\psi\>\<\psi|
&=
\sum_{j,k=0}^n
a_j a_k
(\alpha + (-1)^j\beta )
(\alpha^* + (-1)^k\beta^* )
\notag
\\
&\quad \times |H^{N}_{gj+\Delta}\>
\<H^{N}_{gk+\Delta}|.
\end{align}
Since $gj+\Delta$ satisfy the requirements of Lemma \ref{lem:intermediate-partial-trace}, we get 
\begin{align}
\tr_t(|\psi\>\<\psi|)
    &=
 \sum_{j,k=0}^n
a_j a_k
(\alpha + (-1)^j\beta )
(\alpha^* + (-1)^k\beta^* )
\notag
\\
&\quad \times    
      \sum_{s=0}^t 
 \binom t s    |H^{N-t}_{gj+\Delta-s}\>  \<H^{N-t}_{gk+\Delta-s}|. 
\end{align}
Substituting $b_{j,s} = a_j   \sqrt{\binom {N-t}{gj+\Delta-s}}$,
we can see that 
\begin{align}
 \tr_t(|\psi\>\<\psi|)
    &=
\sum_{s=0}^t \binom  t s
 \sum_{j,k=0}^n
b_{j,s} b_{k,s}
(\alpha + (-1)^j\beta )
\notag\\
&\quad \times
 (\alpha^* + (-1)^k\beta^* )
|D^{N-t}_{gj+\Delta-s}\>  \<D^{N-t}_{gk+\Delta-s}|,
\end{align}
and hence
\begin{align}
 \tr_t(|\psi\>\<\psi|)
    &=
\sum_{s=0}^t \binom t s
|\psi_s\>\<\psi_s|,
\end{align}
from which the result follows.
\end{proof}
In the above lemma, we see that after $|\psi\>$ experiences $t$ deletions, it becomes a convex combination of $t+1$ disjoint orthonormal states $|\phi_s\>$.
We will show how to decode each $|\phi_s\>$ to  $a|0\>+b|1\>$.
\newline

\noindent{\bf Quantum phase estimation:-} The first step of the decoding scheme projects the mixed state 
$\mathcal D_t (|\psi\>\<\psi|)$
onto $|\phi_s\>\<\phi_s|$ and determines the corresponding syndrome $s$. 
For this, we extract the syndrome $s$ using quantum phase estimation on a unitary matrix 
$U = (R^2)^{\otimes N-t}$.
Denoting $w = \exp(2\pi i /g)$ as a $g$th root of unity, we find that 
\begin{align}
    U |\phi_s\> = w^s |\phi_s\>.
\end{align}
To extract the syndrome $s$, we create an ancilla qudit
\begin{align}
    \anc = \sum_{j=0}^{g-1}  |j\>/\sqrt g ,
\end{align}
and apply the controlled operator 
\begin{align}
 V =     \sum_{j=0}^{g-1} |j\>\<j| \otimes U^j
\end{align}
on the state
\begin{align}
    \anc \banc \otimes \mathcal D_t (|\psi\>\<\psi| ).
\end{align}
Since $V$ is a controlled $U^j$ gate, and $U$ is implemented in $N$ steps,
we count the complexity of implementing $V$ as $O(N \log N)$. Next observe that 
\begin{align}
    V \anc\otimes |\phi_s\>
    = \sum_{j=0}^{g-1} \frac{w^{js}}{\sqrt g} |j\> \otimes  |\phi_s\>.
\end{align}
Since the states
$\sum_{j=0}^{g-1} w^{js} |j\>/\sqrt g $ are quantum Fourier transforms of the states $|s\>$, we can obtain the computational basis states $|s\>$ by performing the inverse quantum Fourier transform and measuring in the computational basis $\{|0\>,\dots, |g-1\>\}$ for the ancillary qubit.
Doing so, we obtain the state $|\phi_s\>$ when the state $|s\>$ is obtained on the ancilla. 
The quantum Fourier transform and its inverse are standard quantum algorithms, and can be implemented efficiently \cite{NielsenChuang} in $O((\log g)^2) = o(N)$ steps.
In summary, this quantum phase estimation procedure implements a projective measurement on the symmetric subspace of $N-t$ particles, which filters Dicke states on their weights modulo $g$, and has a complexity of $O(N \log N)$.
\newline

\noindent{\bf Probability amplitude rebalancing:-} 
We now proceed to the second step in the decoding procedure.
The states $|\phi_s\>$ lie in the span of $|+_s\>$ and $|-_s\>$ and are supported on the Dicke states $\{|D^N_{gj + \Delta-s}\> : j=0,\dots, n\}$.
By denoting 
$|0_s\> = (|+_s\> + |-_s\>)/\sqrt 2 $
and 
$|1_s\> = (|+_s\> - |-_s\>)/\sqrt 2 $,
we can see that 
$|+_s\>$ and $|-_s\>$ are supported on 
$|D^N_{gj + \Delta-s}\>$ for which $j$ is even and odd respectively.
We consider a unitary operation $W$ that interchanges the Dicke states 
$|D^N_{gj + \Delta-s}\>$
and $|D^N_{g(n-j) + \Delta-s}\>$.
Note that 
$|0_s\>$ and $|1_s\>$ are not always invariant under the action of $W$.
To rectify this, we rebalance the probability amplitudes on the supports of $|0_s\>$ and $|1_s\>$.
We can achieve this by implementing an appropriate unitary operation $G_s$ that acts non-trivially on only the symmetric subspace, using the method of geometric phase gates which can implement any unitary operation on the symmetric subspace \cite{johnsson2020geometric} in $O(N^2)$ steps. 
This target space that we desire has the basis states 
$|0'_s\>$ and $|1'_s\>$ supported on $|D^N_{g(n-j) + \Delta-s}\>$ for even and odd $j$ respectively, and satisfies the identity $W|0'_s\> = W|1'_s\> $.
Hence if $|\phi_s\> = a|0_s\> + b|1_s\>$, we have
$G_s|\phi_s\> = |\phi'_s\>
= a|0_s'\> + b |1_s'\>$.
\newline

\noindent{\bf Gate teleportation:-} 
The third step of the decoding procedure brings the state $|\phi'_s\> = a|0_s'\> + b |1_s'\>$ to the qubit state $a|0\> + |1\>$ via the gate-teleportation procedure illustrated in Fig. \ref{fig:decoding}.
\begin{figure}[htp]
\centering 
\includegraphics[width=0.35\textwidth]{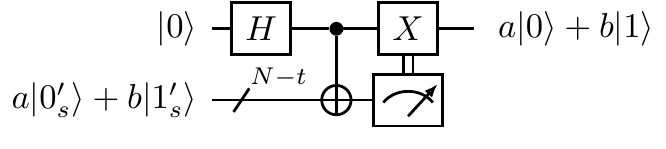}  
\caption{From a rebalanced state
$a|0'_s\> + b|1_s'\>$ with syndrome $s$,
we apply logical gate teleportation to recover the qubit state $a|0\>+ b|1\>.$} 
\label{fig:decoding}
\end{figure} 
To implement the gate teleportation procedure, we prepare an ancilla state in the $(|0\> +|1\>)/\sqrt 2$ state.
Next we perform a logical bit-flip operation on the state $|\phi'_s\>$ conditioned on the value of the ancilla qubit. To achieve this, 
we apply $N-t$ two-qubit CNOT with control qubit on the ancilla and target on each of the qubits in 
$|\phi'_s\>$.
When $N = gn + 2g,$ and $\Delta = g$, the action of this collective CNOT takes the Dicke states
$|D^{N-t}_{gj+g -s}\>$
to 
$|D^{N-t}_{g(n-j)+g -(t-s)}\>$.
Since $s \neq t-s$, this collective CNOT is not yet the logical CNOT on $|\phi'_s\>$. 
To obtain the logical CNOT,  
we need to subsequently apply the collective shift operator $S$ that collectively changes the weight of the entire Dicke basis by $t-2s$, that can be done in $O(N^2)$ time using geometric phase gates \cite{johnsson2020geometric}.
The complexity of this logical CNOT is thus $O(N^2)$.

Next we implement the measurement of the logical $Z$ operation on $|\phi'\>$. This logical $Z$ operation is $R^{\otimes (N-t)}$,
and hence we can measure the observable of $R$ on each qubit. This procedure takes $O(N)$ steps.
Finally, conditioned on the outcome of the logical $Z$ measurement, we apply the qubit bit-flip operation on the ancilla to complete the decoding procedure.

\section{Discussion}
\label{sec:discussion}

In summary, we discussed using permutation-invariant quantum codes to combat quantum deletions. Remarkably, the coding problem for quantum deletions on permutation-invariant quantum codes is equivalent to the coding problem for erasure errors. 
We focused our attention on a special family of permutation-invariant quantum codes, which we call shifted permutation-invariant quantum codes. We furthermore show that  encoding and decoding of such codes can be done in a linear and quadratic number of steps respectively. 
It remains an open problem to optimize the coding schemes presented here, and also to extend these coding schemes beyond the shifted permutation-invariant quantum codes.
It will also be interesting to see how we may correct other types of errors that occur in concert with quantum deletions using the shifted gnu codes proposed here.

 \section{Acknowledgements}\label{eq:acknow}
Y.O. is supported in part by NUS startup
grants (R-263-000-E32-133 and R-263-000-E32-731), and the
National Research Foundation, Prime Minister ’s Office, Singapore
and the Ministry of Education, Singapore under the
Research Centres of Excellence programme.

\bibliography{pidel}{}
\bibliographystyle{ieeetr}

\end{document}